\pdfoutput=1

\documentclass[11pt]{article}
\usepackage[sort]{cite}
\usepackage[letterpaper,margin=1in]{geometry}
\usepackage{amsmath,amsthm}
\usepackage{algorithm}
\usepackage[hidelinks]{hyperref}
\usepackage[capitalize,nameinlink]{cleveref}
\theoremstyle{plain}
\newtheorem{theorem}{Theorem}[section]
\newcommand{\newsiamthm}[2]{\newtheorem{#1}[theorem]{#2}}
\newsiamthm{lemma}{Lemma}
\newsiamthm{corollary}{Corollary}
\newsiamthm{proposition}{Proposition}
\newsiamthm{definition}{Definition}
\numberwithin{equation}{section}
\numberwithin{table}{section}
\numberwithin{figure}{section}
\numberwithin{algorithm}{section}
\usepackage[T1]{fontenc}

\usepackage{soul}
\usepackage{color}
\usepackage{graphicx}
\usepackage{amsfonts}
\usepackage{amssymb}
\usepackage{tikz}
\usepackage{tabularx}
\usetikzlibrary{calc}
\usepackage[noend]{algpseudocode}
\usepackage{xcolor}
\definecolor{c0}{HTML}{641a80}
\definecolor{c1}{HTML}{b73779}
\usepackage{placeins}
\usepackage{float}%
\usepackage{comment}
\usepackage{xspace}
\usepackage{enumitem}

\DeclareMathOperator*{\argmax}{argmax}

\newcommand{\OPT}{\ensuremath{\mathsf{OPT}\xspace}}
\newcommand{\opt}{\ensuremath{\mathsf{opt}\xspace}}

\newsiamthm{problem}{Problem}
\newsiamthm{observation}{Observation}
\newsiamthm{claim}{Claim}

\makeatletter
\newtheorem*{rep@theorem}{\protect\rep@title}
\newcommand{\newreptheorem}[2]{%
\newenvironment{rep#1}[1]{%
 \def\rep@title{#2 \ref{##1}}%
 \begin{rep@theorem}}%
 {\end{rep@theorem}}}
\makeatother
\newreptheorem{theorem}{Theorem}

\newcommand{\defn}[1]{\textbf{\emph{#1}}}

\algnewcommand\algorithmicforeach{\textbf{for each}}
\algdef{S}[FOR]{ForEach}[1]{\algorithmicforeach\ #1\ \algorithmicdo}
\algnewcommand\algorithmicnot{\textbf{not}}
\algdef{S}[IF]{IfNot}[1]{\algorithmicif\ \algorithmicnot\ #1\ \algorithmicthen}

\newcommand{\E}{\mathbb{E}}

\newcommand{\MemberOf}{\ensuremath{\mathrm{MemberOf}}\xspace}
\newcommand{\FreqOf}{\ensuremath{\mathrm{FreqOf}}\xspace}
\newcommand{\SetOf}{\ensuremath{\mathrm{SetOf}}\xspace}

\newcommand{\algrule}[1][.2pt]{\par\vskip.2\baselineskip\hrule height #1\par\vskip.2\baselineskip}

  \usepackage{nth}
  \usepackage{intcalc}

  \newcommand{\cAAAI}[1]{AAAI\ Conference\ on\ Artificial (AAAI)}

\crefname{observation}{Observation}{Observations}
\Crefname{observation}{Observation}{Observations}

\begin{document}

\title{Fast Sparse $\alpha$-Navigable Graph Construction with Improved Resource Augmentation}
\date{}
\author{
    Alex Conway\thanks{Cornell University (\url{me@ajhconway.com}, \url{btl46@cornell.edu})} \and
    Laxman Dhulipala\thanks{University of Maryland (\url{laxman@umd.edu}, \url{rwen1@umd.edu})}\and
    Martin Farach-Colton\thanks{New York University (\url{mlf9579@nyu.edu}, \url{cmusco@nyu.edu}, \url{ys6378@nyu.edu}, \url{torsten.suel@nyu.edu})}\and
    Ben Landrum\footnotemark[1]\and
    Christopher Musco\footnotemark[3]\and
    Yarin Shechter\footnotemark[3]\and
    Torsten Suel\footnotemark[3]\and
    Richard Wen\footnotemark[2] 
}

\maketitle

\begin{abstract}
Graph-based nearest neighbor search methods have surged in popularity in recent years. These methods are based on running greedy search on a directed graph constructed over a set of $n$ points, $x_1, \ldots, x_n$, endowed with some distance function. A common goal is to construct a search graph that has the \emph{$\alpha$-navigability} property, which ensures that greedy search returns provably good approximate nearest neighbors and converges in a small number of iterations. As the parameter $\alpha \geq 1$ increases, accuracy improves, as does convergence rate, but there is a tradeoff: higher values of $\alpha$ require the graph to have more edges, which increases both the data structure's space and search time. This tradeoff has motivated the study of algorithms for constructing a \emph{sparsest $\alpha$-navigable graph} for a given set of points. 

In this work, we present improved algorithms for constructing $\alpha$-navigable graphs. For any point set and any distance function, we show how to compute an $\alpha$-navigable graph $G$ with a number of edges no more than $O(\log n)$ times that of the sparsest $(1+\varepsilon)\alpha$-navigable graph. The algorithm runs in $O(n^\omega \log(\Delta)/\varepsilon)$ time, where $\omega < 2.372$ is the exponent of fast matrix multiplication and $\Delta$ is the ratio between the largest and smallest pairwise distances in the input set. This result improves on recent work by Khanna, Padaki, and Waingarten (SODA 2026), who achieve the same runtime, but are only competitive with the sparsest $2(1+\varepsilon)\alpha$-navigable graph. Since sparsity often increases rapidly with $\alpha$ (typically $\alpha$ is chosen $\leq 1.1$) such a result can be significantly weaker than our bound. Previously, the best algorithm that could compete with the sparsest $(1+\varepsilon)\alpha$-navigable graph required $\widetilde{O}(n^{2.5})$ time. 
Beyond our main result, we show that our techniques extend to the related problem of constructing sparse $\tau$-monotonic graphs, another popular target in nearest neighbor search. Moreover, we provide a fine-grained complexity analysis that suggests that improving the $O(n^{\omega})$ running time in our result is unlikely, or at least would require  a very different algorithmic approach. This analysis is based on proving hardness results for problems related to \emph{verifying} whether a graph is navigable, which is a topic of independent interest. 
\end{abstract}

\medskip\medskip

\section{Introduction}

Nearest neighbor search has been a central problem in computer science for decades \cite{Clarkson:1994,Kleinberg:1997,KushilevitzOstrovskyRabani:1998,Charikar:2002,BeygelzimerKakadeLangford:2006,Clarkson:2006,JegouDouzeSchmid:2011,AndoniIndykNguyen:2014,AumullerBernhardssonFaithfull:2020}.
Given a dataset $P=\{p_1,\ldots,p_n\}$ together with a distance function $d(\cdot,\cdot )$, the goal is to construct a data structure that can answer (approximate) nearest-neighbor queries: given \(q\), return a \(p\in P\) that (approximately) minimizes \(d(p,q)\). 
Interest in this problem has recently surged due to its central role in modern AI systems, which often represent data using high-dimensional vector embeddings.
Tasks such as semantic search, retrieval-augmented generation, recommendation, and clustering all require fast nearest-neighbor search over such vector representations \cite{LewisPerezPiktus:2020,KitaevKaiserLevskaya:2020,XiongXiongLi:2021,LuoLakshmanShrivastava:2022,Bruch:2024}.

Given the problem's importance, many different methods have been proposed to solve approximate nearest neighbor search, including well-known and widely used methods like locality sensitive hashing \cite{IndykMotwani:1998,AndoniIndyk:2008,AndoniIndykLaarhoven:2015}, quantization-based inverted file indices \cite{JegouDouzeSchmid:2011,JohnsonDouzeJegou:2021,GuoSunLindgren:2020}, and tree-based methods \cite{BeygelzimerKakadeLangford:2006,MujaLowe:2009}. In recent years, there has been a surge of interest in a class of algorithms referred to as \defn{graph-based nearest neighbor search methods}. Such methods include Hierarchical Navigable Small World (HNSW) 
Graphs \cite{MalkovYashunin:2020}, Vamana/DiskANN \cite{SubramanyaDevvritKadekodi:2019,KrishnaswamyManoharSimhadri:2024}, 
Navigating Spreading-out Graphs (NSG) \cite{FuXiangWang:2019}, and more \cite{MalkovPonomarenkoLogvinov:2014,FuCai:2016,HarwoodDrummond:2016,PengChoiChan:2023,ManoharShenBlelloch:2024,AumullerBernhardssonFaithfull:2020,SimhadriAumullerIngber:2024,AryaMount:1993}. Such methods offer state-of-the-art performance on a number of leading benchmarks and search competitions \cite{AumullerBernhardssonFaithfull:2020,SimhadriAumullerIngber:2024,SimhadriWilliamsAumuller:2022}.

In graph-based methods, the core data structure used to support nearest neighbor search is a directed graph $G=(P,E)$ with a node corresponding to every point in the dataset. A query $q$ is serviced by starting at some node in the graph, and greedily moving to the out-neighbor that is closest to $q$. The search continues in this way, terminating when no closer neighbor exists.
Given the simplicity of this greedy search algorithm, a key differentiator between graph-based methods is how the search graph $G$ is constructed. How should the edges in $E$ be chosen to facilitate fast and accurate nearest-neighbor search?

\subsection{$\alpha$-Navigable Graph Construction}\label{sec:prior_works}
While many different approaches have been employed to construct search graphs, \defn{navigability} has been repeatedly identified as a desirable property for such graphs \cite{MalkovPonomarenkoLogvinov:2014,SubramanyaDevvritKadekodi:2019}. A graph is \defn{navigable} if, for every source point $s$ and distinct target point $t$ in $P$, the graph has an edge $(s, v)$ where $d(v, t) < d(s, t)$, i.e. the source has an edge that leads closer to the target. On such a graph, a greedy traversal from any starting point in $P$ must converge to any query point, so long as that query is itself a member of $P$; at every step, the navigability property ensures that the traversal can make progress, until it lands on the query itself. 

This basic property can be strengthened to \defn{\(\alpha\)-navigability}, defined for any parameter $\alpha \geq 1$. In this case, it is not sufficient that a source simply have an edge leading closer to each target: it must have an edge that shrinks the distance to a target by a multiplicative factor of \(\alpha\) or more. Formally, for each source-target pair $s,t$, there must be an edge $(s, v)$ where $d(v, t) < \frac{1}{\alpha}d(s, t)$. 

\(\alpha\)-navigability is a powerful property. If $G$ has this property, then greedy search not only converges to any query in $P$, but it can be shown that it provides provably accurate approximate nearest neighbors for arbitrary queries. In particular, as long as the distance function $d$ is a metric, for any query $q$, greedy search will converge to a nearest neighbor whose distance from $q$ is no more than $\frac{\alpha + 1}{\alpha - 1}$ times that of the query's true nearest neighbor \cite{IndykXu:2023}. We obtain a better approximation for larger values of $\alpha$.

Moreover, in addition to ensuring provably accurate results, $\alpha$-navigability ensures fast convergence. When $\alpha > 1$, we will converge to an $\frac{\alpha + 1}{\alpha - 1}$-approximate nearest neighbor in no more than
\[
\lceil\log_\alpha(\Delta_P)\rceil,
\]
steps, where the \defn{aspect ratio} $\Delta_P$ of $P$, is the ratio between the largest and smallest pairwise distances in $P$. Moreover, search does not just converge for queries in $P$. 

Fast convergence alone, however, does not govern the computation cost of graph-based nearest neighbor search. In particular, the cost of every search iteration scales with the number of out-neighbors of the current node. I.e., search is faster if $G$ is sparse. Graph sparsity also  controls the space complexity of the search data structure, an important consideration in applications. The importance of sparsity motivates the main algorithmic question considered in this paper:

\begin{problem}[Sparse $\alpha$-Navigable Graph Construction]
    Given a dataset $P = \{p_1, \dots, p_n\}$, distance function $d(\cdot, \cdot)$, and $\alpha \geq 1$, return an $\alpha$-navigable graph $G = (P, E)$ with as few edges as possible.
\end{problem}

The study of this problem was recently initiated in \cite{ConwayEtAl2025} and \cite{KhannaPadakiWaingarten:2025}. While constructing the sparsest $\alpha$-navigable graph is NP-hard, even for points in Euclidean space \cite{ConwayEtAl2025}, an $O(\log n)$  approximation can be obtained in $O(n^3)$ time by reframing the problem as $n$ set cover instances (see \ref{subsec:set_cover_perspective} for more details). These recent papers improve this bound. For example, letting $\OPT_\alpha$ denote the minimum number of edges in any $\alpha$-navigable graph for $P$, \cite{ConwayEtAl2025} presents an algorithm for constructing an $\alpha$-navigable graph with $O(\OPT_\alpha \log n)$ edges that runs in worst-case $\widetilde O(n^{2.5})$ time.\footnote{For the special case of $\alpha = 1$, this runtime can be further improved to \(\widetilde{O}(n^2)\), which is essentially optimal under the Strong Exponential Time Hypothesis \cite{ConwayEtAl2025}. We also note that for general $\alpha$, the method actually runs in $\widetilde O(\min(n \cdot \OPT_\alpha, n^{2.5}))$, so can be faster than $\widetilde O(n^{2.5})$  if $\OPT_\alpha = o(n^{1.5})$.}

\cite{KhannaPadakiWaingarten:2025} gives an even faster algorithm in a \defn{resource-augmented} setting. They present an $\widetilde O(n^\omega \log(\Delta_P)/\varepsilon)$ time algorithm that computes an $\alpha$-navigable graph with $O(\OPT_{2\alpha(1 + \varepsilon)}\log n)$ edges, where $\varepsilon > 0$ is any small constant and $2 \leq \omega < 2.372$ is the exponent of fast $n \times n$ matrix multiplication. This result is {resource-augmented} bound because $2\alpha(1 + \varepsilon) > \alpha$, we have that $\OPT_{2\alpha(1 + \varepsilon)} \geq \OPT_\alpha$.\footnote{Prior work uses the term \defn{bicriteria} to refer to these algorithms.  We prefer to use the term \emph{resource-augmented} since there are in fact not two criteria being optimized. Rather, we our goal is to return a graph whose size is competitive with the sparsest graph satisfying a slightly stricter set of constraints.} Thus, the algorithm returns a graph whose sparsity is within an $O(\log n)$ factor of an objective larger than $\OPT_\alpha$. Such an algorithm can yield a runtime improvement if $\OPT_{\alpha} \in \widetilde\omega(n^{\omega - 1}\log(\Delta_P))$, as other algorithms exist which run in $\widetilde O(n \cdot \OPT_\alpha)$ time.

\subsection{How much resource augmentation is too much?}
While resource augmentation leads to a clear runtime advantage, reducing graph construction cost from roughly $\widetilde{O}(n^{2.5})$ time to roughly $\widetilde{O}(n^{\omega})$ time, it is natural to ask how much of a price is paid in terms of graph sparsity. In other words, \textbf{how much larger is $\text{OPT}_{2\alpha (1+\varepsilon)}$ than  $\text{OPT}_\alpha$?}

In short, the answer is that the gap between these quantities can be substantial. When constructing $\alpha$-navigable graphs for real-world datasets, practitioners tend to set $\alpha$ close to 1 (usually in the range $[1.0, 1.2])$ \cite{SubramanyaDevvritKadekodi:2019, MalkovYashunin:2020, singh2021freshdiskann, ManoharShenBlelloch:2024}. $\alpha$ is always chosen less than $2$. Recalling that higher $\alpha$ improves both accuracy and convergence rate, this suggests that even $\text{OPT}_2 \leq \text{OPT}_{2\alpha (1+\varepsilon)}$ is too large to be usable. Indeed it has been confirmed experimentally that, on many real-world datasets, as $\alpha$ grows towards 2, the sparsest $\alpha$-navigable graph becomes essentially complete.%

One explanation for this is that, while not truly random, many modern high-dimensional vector embedding datasets exhibit properties of random points in Euclidean space. In particular, the majority of vectors are nearly orthogonal to each other, so distances between points tend to be similar. Indeed, in common embedding datasets, the aspect ratio is quite low, with the maximum pairwise distance no more than a small factor of the minimum. In an $\alpha$-navigable graph with $\alpha > 2$, any source-target pair requires an edge that gets at least $50\%$ closer to the target. For common datasets, such an edge may simply not exist: the only option is to connect directly to the target. I.e., $\OPT_2$ will be close to $n^2$.

Given this, the goal of this work is to develop graph construction algorithms that enjoy the same efficient runtime as the method of \cite{KhannaPadakiWaingarten:2025}, but without sacrificing so much on sparsity. In particular, instead of competing with $\OPT_{2(1+\varepsilon)\alpha}$, ideally we would return a graph whose sparsity is competitive with the best $(1+\varepsilon)\alpha$-navigable graph, where $\varepsilon$ can be set arbitrarily close to zero. 

\subsection{Our Contributions}
The discussion above exposes a concrete gap between existing construction algorithms: one can obtain a strong sparsity guarantee with a running time of roughly $\widetilde O(n^{2.5})$, or a faster $\widetilde O(n^\omega \log(\Delta) / \varepsilon)$ running time with resource augmentation that may make the resulting sparsity guarantee prohibitively weak. Our results largely close this gap. We first give an $\widetilde O(n^\omega \log(\Delta) / \varepsilon)$-time algorithm that reduces the resource augmentation from $2\alpha(1 + \varepsilon)$ to $\alpha(1 + \varepsilon)$. We then study the navigability-verification subroutine that forms the principal bottleneck in this algorithm, giving conditional lower bounds on improving or eliminating it.

\paragraph{Improved resource augmented approximation algorithm}
Our first result is an improvement of \cite{KhannaPadakiWaingarten:2025}'s resource augmented approximation algorithm, tightening the sparseness guarantee.
\begin{theorem}\label{thm:main_upper_bound}
For any $\varepsilon > 0$, there exists a randomized algorithm that, given point set \(P\), arbitrary distance function \(d(\cdot,\cdot)\), and \(\alpha>1\), constructs an \(\alpha\)-navigable graph, with at most
\[
O\big(\text{OPT}_{\alpha(1+\varepsilon)}\log n\big)
\]
edges,\footnote{The guarantee is actually slightly stronger: the constructed graph not only minimizes the total number of edges, but also individual degrees. That is, the out-degree of each point is within $O(\log n)$ of its lowest possible out-degree in any $\alpha$-navigable graph over $P$.} where \(\text{OPT}_{\alpha(1+\varepsilon)}\) is the size of the sparsest \(\alpha(1+\varepsilon)\)-navigable graph. The algorithm runs in \(\widetilde{O}\big(n^\omega \cdot \log (\Delta_P)/\varepsilon \big)\) time,\footnote{We assume $d(\cdot, \cdot)$ takes $O(1)$ time. More generally, if computing $d(\cdot, \cdot)$ takes $O(T)$ time, then our algorithm runs in $\widetilde{O}\left(n^\omega \cdot \log (\Delta_P)/\varepsilon \right) + O(n^2 \cdot T)$ time. The same caveat applies to the prior works discussed in \Cref{sec:prior_works}, which also take on an additional $O(n^2 \cdot T)$ cost for computing distances.} where \(\Delta_P\) is the set's aspect ratio.
\end{theorem}

The starting point for \Cref{thm:main_upper_bound} is the observation that a sparsest $\alpha$-navigable graph problem can be formulated as a combination of $n$ individual set cover problems. Since the $\alpha$-navigable constraint for a source-target pair $(s, t)$ is satisfied by an edge $(s, v)$ where $d(v, t) < \frac{1}{\alpha}d(s, t)$, we can express a source point $s$ as a set cover problem where the elements to be covered are all targets $t \not= s$. In this set cover instance, the sets correspond with potential outedges from $s$, with the set associated with edge $(s, v)$ containing all targets $t$ where $d(v, t) < \frac{1}{\alpha}d(s, t)$. Our main approach, then, is to efficiently solve these $n$ set cover instances.

\cite{KhannaPadakiWaingarten:2025} also makes this observation, but its $\widetilde O(n^\omega \cdot \log(\Delta)/\varepsilon)$ time algorithm does not solve the set cover instances directly, instead adding edges through a geometric argument relying on triangle inequality. At a high level, the driving observation is that an edge from a source point $s$ to an intermediate point $v$ that approaches a given target by a factor of $2\alpha$ can be replaced with an edge to some other point near $v$ and still approach the target by a factor of $\alpha$. Thus, the authors add roughly $\OPT_{2\alpha(1 + \varepsilon)}$ edges and argue that these solve a large number of the source-target constraints from required for an $\alpha$-navigable graph. This is the source of the factor of $2$ from their resource augmentation.

To eliminate the factor of $2$, we replace this with a procedure that directly accesses and solves the underlying set cover instances. This not only comes with a tightening of the resource augmentation but also removes the requirement that the points live in a metric space. Unlike \cite{KhannaPadakiWaingarten:2025}, our algorithm works on any symmetric distance function, without need for triangle inequality. We do, however, reuse a procedure from \cite{KhannaPadakiWaingarten:2025}, which we detail in \Cref{subsec:verify_nav}, that there is a small number of discrete pairwise distances among the points in $P$. To this end, we choose several fixed distance scales and round each distance to one of these, incurring up to a $(1 + \varepsilon)$ distortion on all distances. Thus, we retain the $(1 + \varepsilon)$ resource augmentation factor.

Finally, we show that the same framework extends to the related problem of constructing sparse \(\tau\)-monotonic graphs; see \Cref{subsec:tau_monotonic}.

\paragraph{Hardness of navigability verification.} A key component of our construction, inherited from \cite{KhannaPadakiWaingarten:2025}, is a black-box procedure that determines which source-target $\alpha$-navigability constraints are satisfied by the edges currently in the graph. After distances are discretized, the constraints at each distance scale can be checked using one matrix multiplication. These matrix multiplications dominate the running time of our algorithm, so a substantially faster verification procedure would immediately yield a faster resource-augmented constructions.

This motivates us to study navigability verification as an independent algorithmic problem. We prove two hardness results. First, even when restricted to source-target pairs at a single fixed distance, $\alpha$-navigability verification is at least as hard as Boolean matrix multiplication: an $n \times n$ Boolean matrix product can be computed using one call to the verification procedure and $O(n^2)$ additional time. Thus, bypassing matrix multiplication in our verification routine would require a corresponding breakthrough for Boolean matrix multiplication. Second, we consider exact verification over all undiscretized distances. We show that this problem is equivalent, up to polylogarithmic factors, to Boolean min-plus product. Combining this equivalence with a recent conditional lower bound for Boolean min-plus product \cite{Fischer:2026}, we show that, under the APSP conjecture, exact $\alpha$-navigability verification cannot be solved in $O(n^{2.5 - \delta})$ time for any constant $\delta > 0$. Together, these results suggest that both the matrix multiplication bottleneck and the use of resource augmentation are difficult to avoid within our framework.

\subsection{Technical Overview and Prior Work} \label{subsec:FastConstructions}

A first general answer for the question of how sparse a navigable graph can be made, and how efficiently such a graph can be constructed was given by \cite{DiwanGouMusco:2024}. They showed every \(n\)-point dataset admits a nontrivial sparse navigable graph, independent of the ambient dimension and without relying on metric structure.  Their construction takes the union of a nearest-neighbor graph and a random graph, and in \(\widetilde O(n^2)\) time constructs a navigable graph with \(\widetilde O(n^{3/2})\) edges.  Equivalently, every point set admits a navigable graph of average degree \(\widetilde O(\sqrt n)\).  This bound is essentially tight in the worst-case: for random points in \(O(\log n)\)-dimensional Euclidean space, a navigable graph must have average degree \(\Omega({n^{1/2-o(1)}})\) (with high probability).

This worst-case picture, however, does not capture the instance-by-instance structure of the problem.  Many datasets may admit navigable graphs far sparser than the \(\widetilde O(n^{3/2})\)-edge worst-case guarantee.  Thus the more refined objective is not merely to construct some sparse navigable graph, but to approximate the sparsest navigable graph for the given instance.  This instance-optimal viewpoint was subsequently developed~\cite{ConwayEtAl2025,KhannaPadakiWaingarten:2025}. The key observation in both of these works is that the sparsest $\alpha$-navigable graph problem decomposes into \(n\) coupled set-cover instances, one for each source $s$. For a fixed source \(s\), a candidate outgoing edge \((s,u)\) covers exactly those targets \(t\) for which $u$ gives a valid $\alpha$-progress step toward $t$. Thus, choosing the minimum out-neighborhood of \(s\) is precisely a minimum set-cover problem over the set of targets \(P\setminus \{s\}\). Solving these \(n\) set-cover instances near-optimally yields a near-optimal sparsest $\alpha$-navigable graph. In fact, it produces an even stronger approximation property: for every \(s\in P\), the out-neighborhood of \(s\) is approximately the size of the minimum neighborhood of \(s\) in any navigable graph. This set-cover formulation immediately yields an \(O(\log n)\)-approximation to the optimum via the greedy set-cover algorithm, but a direct implementation takes cubic time even only to construct the sets explicitly.

The main algorithmic driver of \cite{ConwayEtAl2025} is a randomized simulation of the greedy algorithm, intended to solve each set cover instance without explicit constructions of the sets. At each step, a universe element is sampled. Then, all sets covering the element are recorded. The amount of votes a set received is used as an estimator for its size, used to pick larger sets. This produces an algorithm that computes a \(O(\log n)\)-approximation of the sparsest navigable graph in \(\widetilde{O}(n\cdot \text{OPT})\) time. Pairing this with clever preprocessing that exploits relationships between the different set cover instances reduces the running time for general \(\alpha\geq1\) to \(\widetilde{O}(\min\{n\cdot \text{OPT},n^{2.5}\}\big)\). For the case of \(\alpha=1\), i.e., navigable graphs, further dependencies can be exploited to produce a near-optimal \(\widetilde{O}(n^2)\) running time.

The work by \cite{KhannaPadakiWaingarten:2025} includes an algorithm that computes an \(O(\log n)\)-approximation of the sparsest \(\alpha\)-navigable graph in \(\widetilde{O}(n\cdot \text{OPT})\) time. It features the same core idea of a randomized simulation of the greedy algorithm seen in \cite{ConwayEtAl2025}. It also includes a potentially faster resource augmented approximation algorithm that runs in \(O(n^\omega\cdot \log(\Delta)/\varepsilon)\) time, and computes an \(O(\log n)\)-approximation of the sparsest \(2\alpha\cdot(1+\varepsilon)\)-navigable graph, where \(\varepsilon>0\) can be chosen arbitrarily, and \(\Delta\) is the aspect ratio of the point set. Observe that this is a weaker density guarantee as the sparsest \(2\alpha\cdot(1+\varepsilon)\)-navigable graph may be much denser than the sparsest \(\alpha\)-navigable graph. This resource augmented approximation is largely comprised of two subroutines. The first subroutine covers a constant fraction of all uncovered source-target constraints, and the second subroutine is used to verify all satisfied constraints. The implementation of the first subroutine introduces the multiplicative \(2\) in the \(2\alpha\cdot(1+\varepsilon)\)-approximation. Further, for a fast implementation of the second subroutine, specifically by using fast matrix multiplication, distances are discretized, which introduces the multiplicative \((1+\varepsilon)\) approximation term.

\section{Preliminaries}

\subsection{Navigability and \(\alpha\)-Navigability} \label{subsec:NavDefinitions}
\begin{definition}[Navigable Graph]
    \label{def:navigability}
Given point set $P$ and distance function $d : P\times P \to \mathbb{R}^+$, we say that a directed graph $G=\left(P,E\right)$ is \emph{navigable} if, for any source point $s \in P$ and target point $t \in P$, $s \not= t$, there exists some edge $(s, v) \in E$ such that $d(v, t) < d(s, t)$. That is, $s$ has an edge that, when taken, gets closer to $t$.
\end{definition}
Equivalently, for every source \(s\) and target \(t\), there exists a directed path \(p_0=s,p_1,...,p_k=t\) from \(s\) to \(t\) such that \(\forall i:d(p_{i+1},t)<d(p_i,t)\). That is for every source-target pair, there exists a monotone path from the source to the target, and thus the most basic form of greedy search will succeed. Navigability has been studied in the nearest-neighbor search literature for decades \cite{JaromczykToussaint:1992,AryaMount:1993, KhannaPadakiWaingarten:2025,ConwayEtAl2025}. 

While the definition of navigability only guarantees the convergence of basic greedy search for targets that are within \(P\), recent work also shows that multiplicative approximation guarantees can be obtained for points outside of \(P\) in an arbitrary metric when using Adaptive Beam Search, a variant of basic greedy search \cite{Al-JazzaziDiwanGou:2025}.
\paragraph{\(\alpha\)-Navigable graphs} A well-studied direct generalization of navigable graphs  \cite{KrishnaswamyManoharSimhadri:2024,IndykXu:2023, KhannaPadakiWaingarten:2025, ConwayEtAl2025}, \(\alpha\)-Navigability requires not only that a monotone path between every source and target exists, but that the distance to the target shrinks at a multiplicative rate of \(\frac{1}{\alpha}\) at every step along the path.
\begin{definition}[\(\alpha\)-Navigable Graph]
    \label{def:navigability}
Given point set $P$ and distance function $d : P\times P \to \mathbb{R}^+$, and \(\alpha\geq 1\), we say that a directed graph $G=\left(P,E\right)$ is \emph{\(\alpha\)-navigable} if, for any source point $s \in P$ and target point $t \in P$, $s \not= t$, there exists some edge $(s, v) \in E$ such that $d(v,t) < \frac{1}{\alpha}\cdot  d(s,t)$. 
\end{definition}
A ubiquitous modern nearest-neighbor search algorithm, DiskANN \cite{KrishnaswamyManoharSimhadri:2024} constructs graphs that are \(\alpha\)-navigable \footnote{Since the algorithm producing a strict \(\alpha\)-navigable graphs is too slow in practice, the authors of \cite{KrishnaswamyManoharSimhadri:2024} propose a faster heuristic method, which is widely used.}. While the definition of navigability maintains that every greedy search converges, \(\alpha\)-navigable graphs aim to ensure fast convergence by imposing a rate of progress. Further, recent work \cite{IndykXu:2023} shows that greedy search over \(\alpha\)-Navigable graphs always returns an approximate nearest neighbor $\widetilde{p}$ with  $d(q,\widetilde{p}) \leq \frac{\alpha+1}{\alpha-1} \cdot \min_{p\in P} d(q,p)$, even when \(q\notin P\). 

\subsection{Distance Access Model}\label{subsec:distance_access_model}

Our input is a set of points $P = \{p_1, \dots, p_n\}$ together with oracle access to a distance function $d : P \times P \to \mathbb{R}_{\geq 0}$. A query to the oracle returns the distance between two points in $O(1)$ time. We otherwise treat the points and the distance function as black boxes; in particular, our algorithms do not require access to an explicit representation of the points.

We make only two assumptions on the distance function. The first is that $d$ is symmetric, that is, $d(p_i, p_j) = d(p_j, p_i)$.\footnote{Technically, this is not strictly required since all distances during greedy traversal are computed against a common query point, so traversal converges even in the presence of asymmetric distances.} The second is that $d(p_i, p_j)$ is 0 if $p_i = p_j$ and positive if $p_i \not= p_j$. This is so that a point is its own nearest neighbor; otherwise, greedy traversal may not converge on a query even on a navigable graph. Notably, we do not require that $d$ is a metric.

\subsection{Reducing Sparsest Navigable Graphs to Minimum Set Cover}\label{subsec:set_cover_perspective}
It has been observed independently by \cite{ConwayEtAl2025,KhannaPadakiWaingarten:2025}, that the problem of constructing a sparse $\alpha$-navigable graph over a point set $P=\{p_1,...,p_n\}$ can be cast as a collection of $n$ coupled set cover instances, each with $n-1$ elements, and $n-1$ sets. Specifically, we have a single set cover instance for each point. The instance corresponding to point $p_i$ involves covering all points $p_k\neq p_i$, and has a set for every $j\neq i$, denoted by $S_{i\to j}$. $S_{i\to j}$ contains every point $p_j$ provides a greedy step towards when standing at $p_i$, i.e., all points $p_k$ such that $d(p_j,p_k) < \frac{1}{\alpha}\cdot d(p_i,p_k)$. In other words, if we consider source-target $\alpha$-navigability constraints in which $p_i$ is the source, $S_{i\to j}$ contains all targets which corresponding constraints are satisfied by adding the directed edge $(p_i,p_j)$. Observe that a set cover for this instance corresponds to an out-neighborhood for $p_i$ that covers all $\alpha$-navigability constraints in which $p_i$ is the source, and a minimum set cover corresponds to a minimum such out-neighborhood. Consequently, finding minimum set cover for all said instances produces a minimum navigable graph. The relationship between these different set cover instances was leveraged by \cite{ConwayEtAl2025} to produce efficient query access to these instances without explicitly constructing them, which we will review next.

\subsection{Set Access Model and Efficient Query Access for Set Cover}\label{subsec:set_access_model}
Under the distance access model described in \Cref{subsec:distance_access_model}, we can efficiently query the set cover instances induced by the problem of constructing a sparse navigable graph over $P$.
In particular, within the set cover instance corresponding to $p_i$, \cite{ConwayEtAl2025} showed that with $O(n^2 \log n)$ preprocessing, we can support the following three operations in constant time:
\begin{enumerate}[label=]
    \item $\MemberOf(p_k, S_{i \to j})$ returns \textsc{Yes} if $p_k \in S_{i \to j}$ and \textsc{No} otherwise.
    \item $\FreqOf(p_k)$ returns the number of sets $S_{i \to j}$ containing $p_k$.
    \item $\SetOf(p_k, \ell \in [n])$ returns the $\ell$th set $S_{i \to j}$ containing $p_k$ (in some order), assuming $\ell \leq \FreqOf(p_k)$.
\end{enumerate}
Suppose we fix the set cover instance, that is, we fix $i$. Supporting \MemberOf queries is straightforward. By definition, $p_k \in S_{i \to j}$ if and only if $d(p_j,p_k) < \frac{1}{\alpha} \cdot d(p_i,p_k)$, which can be determined by only two distance accesses.

To support the other two queries, we first preprocess a \defn{distance-based permutation matrix}. This is an $n \times n$ matrix $\Pi$ where the $i$th row stores the sorted order of all points by increasing distance from $p_i$. That is, $p_{\Pi_{i, j}}$ is the $j$th closest point to $p_i$. This matrix can be constructed in $O(n^2 \log n)$ time by simply sorting each row.

Using the permutation matrix, we observe that, because $p_k \in S_{i \to j} \iff d(p_j, p_k) < \frac{1}{\alpha} \cdot d(p_i, p_k)$, all of the sets $S_{i \to j}$ containing $p_k$ are associated with the closest points to $p_k$ and can therefore be found within some prefix of the $k$th row of $\Pi$. We can find the length of the prefix by binary search on row $k$ for the last point $j$ satisfying $d(p_j, p_k) < \frac{1}{\alpha} d(p_i, p_k)$. We can precompute these prefix lengths for all instances $i$ and targets $p_k$ in $O(n^2 \log n)$ time total, then perform a simple lookup of these prefix lengths to answer \FreqOf queries. \SetOf queries follow the same principle; the $\ell$th set containing $p_k$ is simply $S_{i \to \Pi_{k, \ell}}$.

From this point on, we assume we have access to these three operations. We state this as a lemma.
\begin{lemma}[due to \cite{ConwayEtAl2025}]
    Given a sparsest $\alpha$-navigable graph construction problem specified by an input point set $P$ and distance function $d:P\times P\to \mathbb{R}^+$, there exists a data structure which supports \MemberOf, \FreqOf, and \SetOf queries for the induced set cover instances in $O(1)$ time. This data structure is constructed in $O(n^2\log n)$ time and requires $O(n^2)$ space.
\end{lemma}

\subsection{Batched Verification of Navigability Constraints via Matrix Multiplication}\label{subsec:verify_nav}
As in \cite{KhannaPadakiWaingarten:2025}, a key component of our new algorithm for constructing sparse $\alpha$-navigable graphs is a routine for using fast matrix multiplication to verify the $\alpha$-navigability of a graph. Concretely, given a graph over $P$, we'd like to identify which source-target $\alpha$-navigability constraints are satisfied. We first describe how to verify all constraints at a single distance scale, then use distance discretization to verify all constraints.

\begin{lemma}[due to \cite{KhannaPadakiWaingarten:2025}]\label{lem:verify_nav:single_scale}
    Let $G = (P, E)$ be a directed graph, $\alpha \geq 1$, and $r > 0$. In $O(n^\omega)$ time, we can identify every pair $(s, t)$ such that there exists some edge $(s, v) \in E$ satisfying $d(v, t) < r$.
\end{lemma}

\begin{proof}
    Let $A$ be the adjacency matrix of $G$, and define the matrix $B^{(r)} \in \{0, 1\}$ by
    \[
        B^{(r)}_{j, k} = \begin{cases}
            1 & \text{if } d(p_j, p_k) < r\\
            0 & \text{otherwise}
        \end{cases}
    \]
    We construct $B^{(r)}$ using $O(n^2)$ distance queries and compute the matrix product $C^{(r)} = AB^{(r)}$. For any $(i, k)$, $C^{(r)}_{i, k} > 0$ if and only if there exists some $j$ such that $(p_i, p_j) \in E$ and $d(p_j, p_k) < r$. Consequently, we can identify all pairs described in the lemma above by scanning through $C^{(r)}$ and outputting $(p_i, p_k)$ whenever $C^{(r)}_{i, k} > 0$.

    Constructing the matrices and inspecting the product take $O(n^2)$ time, while computing the product takes $O(n^\omega)$ time. Since $\omega \geq 2$, the total running time is $O(n^\omega)$.
\end{proof}

The preceding lemma applies only to pairs at one particular distance. Applying this separately to up to $\Theta(n^2)$ distinct distances is too expensive. We therefore group distances into geometrically increasing scales. For $0 < \varepsilon < 1$, define the distance scales $r_\ell = d_{min}(1 + \varepsilon)^\ell$ for $\ell = 0, \dots, \lceil\log_{1 + \varepsilon}(\Delta)\rceil$, where $d_{min} = \min_{p \not= q} d(p, q)$. We now round each pairwise distance in $P$ down to the nearest scale $r_\ell$, that is, we assign each distance $x \in [d_{min}, d_{max}]$ to the largest scale $r_\ell$ satisfying $r_\ell \leq x$. Consequently, there are only $O(\log_{1 + \varepsilon}(\Delta)) = O(\log(\Delta) / \varepsilon)$ distinct distances, and each new distance differs from its original distance by no more than a factor of $(1 + \varepsilon)$. Let $\rho(x)$ denote the rounded distance for a distance $x$.

We now define the verification procedure used throughout this paper.

\newcommand{\VerifyNav}{\textsc{VerifyNav}\xspace}
\begin{lemma}[Batched verification]
\label{lem:verify_nav}
Given \(P\), a directed graph \(G=(P,E)\), \(\alpha\geq 1\), and \(0<\varepsilon\leq 1\), there is a procedure, \(\VerifyNav(G,\alpha,\varepsilon)\), which runs in
\[
    O\left(
        n^\omega
        \left(1+\log_{1+\varepsilon}\Delta\right)
    \right)
\]
time and returns a set \(\mathcal{N}(G)\) of source-target pairs that
\begin{enumerate}
    \item Includes all $(s, t)$ pairs where no edge $(s, v) \in E$ satisfies $d(v, t) < \frac{1}{\alpha}d(s, t)$
    \item Excludes all $(s, t)$ pairs where some edge $(s, v) \in E$ satisfies $d(v, t) < \frac{1}{\alpha(1 + \varepsilon)}d(s, t)$
\end{enumerate}
Equivalently, if \(\mathcal{U}_\beta(G)\) denotes the set of
unsatisfied \(\beta\)-navigability constraints, then
\[
    \mathcal{U}_\alpha(G)
    \subseteq
    \mathcal{N}(G)
    \subseteq
    \mathcal{U}_{\alpha(1+\varepsilon)}(G).
\]
\end{lemma}

\begin{proof}
    Recall from \Cref{lem:verify_nav:single_scale} that we can identify in $O(n^\omega)$ time the exact set of $(s, t)$ pairs where an edge $(s, v) \in E$ satisfies $d(v, t) < r$.
    
    We execute this procedure on the distance scales $r^\ell / \alpha$ for $\ell = 0, \dots, \lceil\log_{1 + \varepsilon}(\Delta)\rceil$. Thus, we compute, for each distance scale $r^\ell$, the set of all $(s, t)$ pairs where some $(s, v) \in E$ satisfies $d(v, t) < r^\ell / \alpha$. Now take any $(s, t)$ pair. Recall that its rounded distance $\rho(d(s, t))$ matches some distance scale $r^\ell$. For that $r^\ell$, consult whether some $(s, v) \in E$ satisfies $d(v, t) < r^\ell / \alpha$. If not, return $(s, t)$. We claim that doing this over all $(s, t)$ pairs returns all pairs whose $\alpha$-navigability constraint is not satisfied and no pairs whose $\alpha(1 + \varepsilon)$-navigability constraint is satisfied.
    
    First, suppose some $(s, t)$ pair's $\alpha$-navigability constraint is not satisfied. That is, no $(s, v)$ edge satisfies $d(v, t) < \frac{1}{\alpha} d(s, t)$. Since $\rho(d(s, t)) \leq d(s, t)$, it must also hold that no $(s, v)$ edge satisfies the even stronger condition of $d(v, t) < \frac{1}{\alpha} \rho(d(s, t))$. Thus, $(s, t)$ is returned.
    
    Second, suppose some $(s, t)$ pair's $\alpha(1 + \varepsilon)$-navigability constraint is satisfied. That is, there is some edge, $(s, v^*) \in E$, where $d(v^*, t) < \frac{1}{\alpha(1 + \varepsilon)} d(s, t)$. Since rounded distances are less than a $(1 + \varepsilon)$ factor off from true distances, it follows that $\rho(d(s, t)) \geq \frac{1}{1 + \varepsilon} d(s, t)$. Substituting $d(s, t) \leq (1 + \varepsilon) \rho(d(s, t))$ yields that $d(v^*, t) < \frac{(1 + \varepsilon)}{\alpha(1 + \varepsilon)} \rho(d(s, t)) \leq \frac{1}{\alpha} \rho(d(s, t))$. Thus, $(s, t)$ is not returned.
    
    Thus, we have shown that the set of returned source-target pairs satisfies the properties listed in the above lemma. Finally, the entire procedure takes $O(n^2)$ time to check whether to return each pair, as well as $O(\log_{1 + \varepsilon}(\Delta)) = O(\log(\Delta) / \varepsilon)$ calls to the procedure from \Cref{lem:verify_nav:single_scale}, taking $O(n^\omega \log(\Delta) / \varepsilon)$ time total.
\end{proof}

\newcommand{\VoteCover}{\textsc{VoteCover}\xspace}
\newcommand{\PartialCover}{\textsc{PartialVoteCover}\xspace}
\newcommand{\BatchCover}{\textsc{BatchedVoteCover}\xspace}
\section{Fast Resource Augmented Sparsest $\alpha$-Navigable Graphs}

In this section, we prove our main positive result. Recall from \Cref{subsec:set_cover_perspective} that for each source point $s \in P$, choosing the outgoing neighbors of $s$ is exactly a set cover problem: each possible edge $(s, v)$ covers the targets $t$ for which $d(v,t) < \frac{1}{\alpha} d(s,t)$. Thus, constructing a sparse $\alpha$-navigable graph amounts to solving $n$ set cover instances, one for each possible source.

A direct greedy set cover implementation is too slow. Across all sources, there are $n$ set cover instances, each with $n - 1$ candidate sets and $n - 1$ elements. Explicitly materializing all sets already takes $\Omega(n^3)$ time and space in the worst case. \cite{ConwayEtAl2025} circumvents this cost by avoiding building these sets explicitly. However, it incurs the cost of repeatedly updating the set of uncovered targets. For a source $s$, these updates can take $\widetilde \Omega(n \cdot \OPT_\alpha(s))$ time, and summing over all sources gives the $\widetilde \Omega(n \cdot \OPT_\alpha)$-type bottleneck of prior approaches without resource augmentation.

Our main idea is to separate the tasks of proposing a cover for a set of uncovered points and applying that cover to remove those points. Locally, for each source $s$, we use a voting procedure to choose a batch of promising outgoing edges. Globally, after many such edges have been added across many sources, we use fast matrix multiplication (\VerifyNav) to recompute which source-target constraints remain uncovered. Thus, the expensive materialization step is not performed after each selected edge, but only after a whole batch of updates. We will show that, for each set cover instance, each batch of edges added covers a constant fraction of the remaining constraints, so only $O(\log n)$ such global updates are needed. This is similar to the approach taken by \cite{KhannaPadakiWaingarten:2025}, which also sought to repeatedly cover a constant fraction of constraints, but the method we use to achieve this bears closer resemblance to an algorithm described in \cite{IndykMahabadiRubinfeld:2018}, which computes covers over subsamples of a larger set cover instance.

Recall that the \VerifyNav procedure requires rounding distances into a small number of geometric buckets, after which it can no longer distinguish between distances within a bucket, up to a $(1 + \varepsilon)$ multiplicative factor apart. It can, however, distinguish constraints that are definitely unsatisfied at parameter $\alpha$ from constraints that are definitely satisfied at the stronger parameter $\alpha(1 + \varepsilon)$. This is why our sparsity bound is measured against the optimal $\alpha(1 + \varepsilon)$-navigable graph: the residual instances we continue to solve may erroneously contain points falsely believed by \VerifyNav to be unsatisfied at parameter $\alpha$, and it may not be until the graph is $\alpha(1 + \varepsilon)$-navigable before it correctly indicates that the graph is fully $\alpha$-navigable.

For clarity, we have the following local notation. For $\beta \geq 1$ and source $s \in P$, let
\begin{equation*}
    \OPT_\beta(s) = \min\{|\mathcal{C}| : \mathcal{C} \subseteq P \setminus \{s\} \text{ and } \forall t \in P \setminus \{s\}, \exists v \in \mathcal{C} \text{ such that } d(v, t) < d(s, t) / \beta\} \tag{1}
\end{equation*}
Equivalently, $\OPT_\beta(s)$ is the minimum out-degree needed at $s$ to satisfy all $\beta$-navigability constraints whose source is $s$. The theorem below gives a per-source guarantee, and therefore also implies the corresponding global edge bound of \Cref{thm:main_upper_bound} by summing over all sources.

\begin{theorem}\label{thm:upper_bound}
    Let $P$ be a set of $n$ points with distance function $d$ and aspect ratio $\Delta$, and let $\alpha \geq 1$. For any $\varepsilon > 0$, there is a randomized algorithm running in $\widetilde O(n^\omega \log (\Delta)/\varepsilon)$ time that constructs an $\alpha$-navigable graph over $P$ such that, with high probability, every point $s \in P$ has out-degree $\leq O(\log n) \cdot \OPT_{\alpha(1 + \varepsilon)}(s)$.
\end{theorem}

Before giving the algorithm, we spell out the two ingredients that are combined in the proof.

\paragraph{Voting as approximate greedy set cover}
The procedure \VoteCover, introduced in \cite{ConwayEtAl2025}, can be viewed as a fast randomized simulation of the standard greedy set cover algorithm. The standard greedy algorithm achieves a $\log n$ approximation, but naively applied, leads to $O(n^3)$ time to find a near-optimal $\alpha$-navigable graph. Given a source $s$ and a set $U \subseteq P \setminus \{s\}$ of currently uncovered targets, \VoteCover repeatedly samples uncovered targets from $U$. A sampled target $t$ casts a vote for every candidate edge (i.e., set) $(s, v)$ that would cover it, namely every $v$ satisfying $d(v, t) < d(s, t) / \alpha$. Sets with many votes are good estimators for sets covering many currently uncovered targets, so the algorithm adds the set with the most number of votes (after $\sim \log n$ samples), then explicitly updates $U$ by removing every target now covered by that set.

This explicit update of $U$ is the bottleneck. Without efficient access to the points covered by a given edge, this is done in \VoteCover through a full scan over $U$. Because the size of $U$ does not shrink quickly enough with each new edge, this takes $\Omega(n \cdot \OPT_\alpha(s))$ total time in the worst case. Summed over all sources, this is too expensive for the desired $n^\omega$-time bound -- recalling that $\OPT_\alpha = \sum_{s=1}^n \OPT_\alpha(s)$, we would obtain a running time of $\Omega(n \cdot \OPT_\alpha)$. It was shown in \cite{ConwayEtAl2025} that a preprocessing step could reduce this cost to amortized $\widetilde O(n)$ per source by shrinking the initial size of $U$ when $\alpha = 1$, but the approach does not extend in general to $\alpha > 1$.

\paragraph{Batched verification by matrix multiplication}
To reduce the cost of updating $U$, our algorithm introduces a second main ingredient: the verification routine of \Cref{subsec:verify_nav}.

Suppose first that the distances between input points take only $W$ distinct values. This can be enforced by rounding each continuous distance to the nearest discrete distance scale. For a fixed distance scale, the set of constraints satisfied by the current graph can be found by one Boolean matrix multiplication, as reviewed in \Cref{subsec:verify_nav}. Repeating over all $W$ scales gives a verification procedure running in $\widetilde O(n^\omega W)$ time. By rounding distances geometrically to powers of $(1 + \varepsilon)$, we may take $W = O(\log(\Delta) / \varepsilon)$; this is the source of both the $(1 + \varepsilon)$ resource augmentation slack and the $\log(\Delta) / \varepsilon$ factor in the running time.

Our algorithm uses this routine as a global replacement for the expensive scans inside \VoteCover. Rather than updating the uncovered set after each selected edge, we let every source add a batch of edges, and then use a \VerifyNav routine to refresh all their uncovered sets simultaneously. The rest of this section formalizes this batched voting procedure. In \Cref{sec:partial_cover}, we define \PartialCover, which covers a constant fraction of a given residual instance. In \Cref{sec:batch_cover}, we show how to obtain a fast overall runtime by running this partial cover routine for all sources in parallel and using \VerifyNav to update the residual constraints between batches.

\subsection{Partial Cover}\label{sec:partial_cover}

Here we introduce \PartialCover (\Cref{alg:partial_vote_cover}), which differs from \VoteCover in one fundamental way: given a set $U$ of uncovered elements, \PartialCover is only asked to cover at least $|U|/2$ elements of $U$ rather than all of them. 
The way this manifests itself is that \PartialCover constructs a cover over a sample $K$ of $U$ instead of over $U$ directly. Controlling the size of the sample allows us to bound both the amount of work we do and the number of sets in the partial cover. This approach bears similarity to sampling methods used in prior work on set cover algorithms in sublinear-time models of computation \cite{IndykMahabadiRubinfeld:2018}.

The voting process is largely as in \VoteCover. Each uncovered sample casts votes for the sets containing it until some set receives $c \ln n$ votes, where $c$ is a fixed constant. That set is then added to the cover, and the points that voted for it revert their votes for other sets. This repeats until all samples have voted, and the cover is returned.

\begin{algorithm}[htb]
    \caption{\PartialCover($\mathcal{S}_\beta, U, \opt, C, c$)}
    \label{alg:partial_vote_cover}
    \begin{algorithmic}[1]
        \Require A set system $\mathcal{S}_\beta$, a set of elements $U$ to cover, a guess $\opt$ of $\OPT_\beta(s)$, and constants $C, c > 0$.
        \Ensure A partial cover $\mathcal{C} \subseteq \mathcal{S}_\beta$ that covers at least $|U|/2$ elements from $U$.
        \algrule
        \State $\mathcal{C}, V \gets \varnothing$
        \State $K \gets$ a multiset of $C \cdot \opt \ln n$ independent samples from $U$
        \While{$K \not= \varnothing$}
            \State $v \gets$ element sampled uniformly at random from $K$
            \State $K \gets K \setminus v$
            \If{$v \in \bigcup_{S \in \mathcal{C}} S$}\label{alg:partial_cover:check_if_covered} \Comment{Check if any sets containing $v$ (via \SetOf) are in $\mathcal{C}$}%
                \State \textbf{continue}
            \EndIf
            \State $V \gets V \uplus \{v\}$\label{alg:partial_cover:v_votes}%
            \State $S^* \gets \argmax_{S \in \mathcal{S}_\beta} |V \cap S|$\label{alg:partial_cover:find_largest_set}
            \If{$|V \cap S^*| \geq c \ln n$}\label{alg:partial_cover:vote_threshold}
                \State $\mathcal{C} \gets \mathcal{C} \cup \{S^*\}$
                \State $V \gets V \setminus S^*$\label{alg:partial_cover:undo_votes} \Comment{Remove all votes cast by elements covered by $S^*$, including duplicates}
            \EndIf
        \EndWhile
        \State \Return $\mathcal{C}$
    \end{algorithmic}
\end{algorithm}

\subsubsection{Correctness and success probability}
\smallskip

We now prove that \PartialCover covers at least half of $U$ with high probability. First, we note a few simple properties of the output of \PartialCover.

\begin{observation}\label[observation]{obs:partial_cover_output_size}
    \PartialCover returns a partial cover $\mathcal{C}$ containing no more than $\frac{C}{c} \cdot \opt$ sets.
\end{observation}

\begin{proof}
    Each set added to the cover in \PartialCover consumes $c \ln n$ voters. That is, the $c \ln n$ samples that voted for that set are removed from $V$ and their votes eliminated from all sets. Since the initial number of samples in $K$ is $C \cdot \opt \ln n$, the algorithm can successfully find a set exceeding the vote threshold at most $\frac{C \opt \ln n}{c \ln n} = \frac{C}{c} \cdot \opt$ times.
\end{proof}

\begin{observation}\label[observation]{obs:partial_cover_intermediate_num_votes}
    Before or after any iteration of the while-loop of \Cref{alg:partial_vote_cover}, every set $S \in \mathcal{S}$ satisfies $|V \cap S| < c \ln n$.
\end{observation}

\begin{proof}
    We prove this by induction. It is trivially true at the start of the algorithm since $V$ is empty. Now assume that it holds at the start of a given iteration, i.e. $|V \cap S| < c \ln n$ for all $S \in \mathcal{S}$. Observe that on each iteration, $V$ grows by at most one element, $v$, so any $|V \cap S|$ likewise grows by at most one, up to a maximum of $c \ln n$. If any sets $S$ now satisfy $|V \cap S| = c \ln n$, then Line~\ref{alg:partial_cover:vote_threshold} triggers, and one such set is subtracted from $V$. However, all satisfying sets share the element $v$, so this subtraction necessarily removes at least one element, $v$, from all their intersections with $V$. Therefore, all sets once again satisfy $|V \cap S| < c \ln n$.
\end{proof}

\begin{observation}\label[observation]{obs:partial_cover_uncovered_size}
    Of the points initially in $K$, the partial cover $\mathcal{C}$ returned by \PartialCover leaves fewer than $c \cdot \OPT_\beta(s) \ln n$ uncovered. 
\end{observation}

\begin{proof}
    At the end of \PartialCover, the points initially in $K$ that are not covered by $\mathcal{C}$ are exactly those that remain in $V$. By definition of $\OPT_\beta(s)$, the elements of $V \subseteq U$ can be fully covered by at most $\OPT_\beta(s)$ sets, i.e. by the optimum cover of $U$. However, by \Cref{obs:partial_cover_intermediate_num_votes}, each of the sets in this optimum cover contains fewer than $c \ln n$ of the elements in $V$. Thus, the number of points in $V$ is bounded by $|V| < c \cdot \OPT_\beta(s) \ln n$ via the Pigeonhole Principle.
\end{proof}

Using these observations, we now prove that an appropriate choice of the constants $C$ and $c$ ensures that \PartialCover covers at least half of the elements in $U$ with high probability.

\begin{lemma}\label{lem:batch_cover_fraction}
    If $\OPT_\beta(s) \leq \opt \leq 2\OPT_\beta(s)$, then there are sufficiently large constants $C$ and $c$ such that the partial cover $\mathcal{C}$ on $s$ returned by \PartialCover (\Cref{alg:partial_vote_cover}) covers at least half of $U$, i.e. $|U \setminus \bigcup_{S \in \mathcal{C}^*} S| \leq |U|/2$, with probability $\geq 1 - 1/n$.
\end{lemma}

\begin{proof}
    Assume $\OPT_\beta(s) \leq \opt \leq 2\OPT_\beta(s)$.

    Let $U^*$ be the subset of $U$ not covered by $\mathcal{C}$. By \Cref{obs:partial_cover_output_size}, there are at most $R = \frac{C}{c}\opt \leq \frac{2C}{c}\OPT_\beta(s)$ sets in $\mathcal{C}$. It follows that $U^*$ is generated by removing up to $R$ sets of elements from $U$. Let $\mathcal{W}$ be the collection of all valid subsets of $U$ that can be formed by removing the elements of at most $R$ sets chosen from $\mathcal{S}$. The number of such configurations is bounded by $$|\mathcal{W}| \leq \sum_{i=0}^R \binom{|\mathcal{S}|}{i} \leq \sum_{i=0}^R \binom{n}{i} \leq n^R$$

    Let $W$ be any such subset of $U$ where $|W| > |U|/2$. Our goal is to show that the probability that $U^* = W$ is very small, as this will imply that \PartialCover is unlikely to have returned a $\mathcal{C}$ that leaves $> |U|/2$ elements uncovered. Because $K$ consists of $C \cdot \opt \ln n$ independent uniform samples from $U$, the expected number of samples falling into $W$ is
    $$\E[|W \cap K|] = \frac{|W|}{|U|} C \cdot \opt \ln n > \frac{C}{2} \opt \ln n \geq \frac{C}{2} \OPT_\beta(s) \ln n$$
    
    Since $U^* = W$ would necessarily imply that $|W \cap K| = |U^* \cap K|$, the probability that $U^* = W$ is bounded above by the probability that $|W \cap K| = |U^* \cap K|$. However, notice that $U^* \cap K$ is exactly the set of samples in $K$ not covered by $\mathcal{C}$. By \Cref{obs:partial_cover_uncovered_size}, the number of such samples must be strictly less than $c \cdot \OPT_\beta(s) \ln n$.
    
    Applying a Chernoff bound, the probability that $|W \cap K|$ deviates so far below its expectation is
    \begin{align*}
    \Pr[|W \cap K| < c \cdot \OPT_\beta(s) \ln n] & \leq \Pr\left[|W \cap K| < \left(1 - \frac{C - 2c}{C}\right) \cdot E[|W \cap K|]\right]\\
    & \leq \exp\left( - \frac{(C - 2c)^2}{4C} \OPT_\beta(s) \ln n \right)
    \end{align*}
    
    We now union bound this failure probability over all possible valid configurations $W \in \mathcal{W}$
    $$\Pr\left[\exists W \in \mathcal{W} : |W| > \frac{|U|}{2} \land |W \cap K| < c \cdot \OPT_\beta(s) \ln n\right] \leq n^R \exp\left( - \frac{(C - 2c)^2}{4C} \OPT_\beta(s) \ln n \right)$$
    
    Substituting $R \leq \frac{2C}{c} \OPT_\beta(s)$, this probability is bounded by
    $$\exp\left( \left( \frac{2C}{c} - \frac{(C - 2c)^2}{4C} \right) \OPT_\beta(s) \ln n \right)$$
    
    We can make the exponent arbitrarily negative by choosing $c$ and $C$. For instance, let $C = 3c$. The exponent coefficient simplifies to $\frac{6}{1} - \frac{c^2}{12c} = 6 - \frac{c}{12}$. By choosing $c \geq 84$, the coefficient is less than or equal to $-1$. Thus, the probability of failure is at most $\exp(-\OPT_\beta(s) \ln n) \leq 1/n$.

    This proves that, with high probability, $U^* \not= W$ for any $W$ with $> |U|/2$, i.e., \PartialCover must terminate having covered at least half of the elements in $U$.
\end{proof}

\subsubsection{Running time}

Having shown that \PartialCover manages to cover at least half of the elements in $U$, we now explain how we can execute it efficiently. In particular, we aim to show that it runs in $O(n\log n)$ time with high probability. Towards that end, we begin with the following intermediate result:

\begin{lemma}\label{lem:partial_cover_running_time}
    Consider running \PartialCover on the set cover instance corresponding with source $s$, and let $\FreqOf$ be as defined in \Cref{subsec:set_access_model} for that instance, i.e. for a given element $v$, $\FreqOf(v)$ is the number of sets in the instance which contain $v$. Then \PartialCover runs in time $O(\sum_{v \in K} \FreqOf(v))$, where $K$ is the random sample chosen in the algorithm.
\end{lemma}

This is sufficient to meet the desired $O(n\log n)$ bound if we can show that each element is contained in few sets; this will be addressed in \Cref{lem:freq_bounding}.

\begin{proof}
    The runtime bottlenecks in \PartialCover are Lines ~\ref{alg:partial_cover:check_if_covered} and \ref{alg:partial_cover:vote_threshold}. Recall from \Cref{subsec:set_access_model} that we can answer in constant time any of the following queries: (1) \MemberOf, which answers whether a given set contains a given element; (2) \FreqOf, which returns the number of sets containing a given element; and (3) \SetOf, which returns the $i$th (according to some arbitrary order) set containing a given element. We will just assume we have access to these primitives, but we refer back to \Cref{subsec:set_access_model} for a description of how these primitives can be supported in the set cover instances of a sparsest $\alpha$-navigable graph instance.
    
    \paragraph{Line~\ref{alg:partial_cover:check_if_covered}: check if an element has been covered.}
    We maintain $\mathcal{C}$ in a hash table.
    To check if element $v$ is covered by $\mathcal{C}$, we scan over all possible sets that contain $v$ (by finding the number of sets with \FreqOf and iterating over them with \SetOf) and, for each one, check if it is contained in the hash table representing $\mathcal{C}$. We return \textbf{true} if any such set is in $\mathcal{C}$. We do one constant-time membership query per set containing $v$, so $\leq \FreqOf(v)$ per sample $v$, and $\leq \sum_{v \in K} \FreqOf(v)$ in total.
    
    \paragraph{Line~\ref{alg:partial_cover:find_largest_set}: find the set with the most votes.}
    For each $S \in \mathcal{S}$, we maintain the intersection $V \cap S$ in a hash table. We update all of these hash tables any time $V$ is modified. $V$ is modified on Lines~\ref{alg:partial_cover:v_votes} and \ref{alg:partial_cover:undo_votes}.
    
    Line~\ref{alg:partial_cover:v_votes} adds $v$ to $V$. The only intersections $V \cap S$ that change are those where $v \in S$. We can find these by calling $\SetOf$ operations on $v$. For each such set $S$, we do a single constant-time insertion into its $V \cap S$ hash table, for a maximum of $\sum_{v \in K} \FreqOf(v)$ across all voters.
    
    Line~\ref{alg:partial_cover:undo_votes} runs after a set $S^*$ has been added to $\mathcal{C}$. It subtracts $S^*$ from $V$. The elements we must remove from $V$ are exactly those in the hash table for $V \cap S^*$. Thus, we iterate over $V \cap S^*$, and for each element $v$ in it, we must remove $v$ from all $V \cap S$ hash tables. But note that the sets $S$ where $v \in V \cap S$ are exactly those where $v \in S$. We can iterate through these by calling $\SetOf$ operations on $v$. For each of these, we do a single constant-time hash table deletion. We can bound the total number of deletions by observing that any $v$ can only be removed from $V$ once. Since each such $v$ is removed from $\FreqOf(v)$ hash tables, we perform a maximum of $\sum_{v \in K} \FreqOf(v)$ deletions over the course of \PartialCover.
    
    Having bounded the cost of maintaining the hash tables, how do we efficiently identify the largest (Line~\ref{alg:partial_cover:find_largest_set}) in order to evaluate the condition on Line~\ref{alg:partial_cover:vote_threshold}? We first note that it is only necessary to identify the most voted set if the condition passes. We know from \Cref{obs:partial_cover_intermediate_num_votes} that at the start of each loop iteration, the condition must be false. Thus, it can only become true when a hash table grows. It follows that it is sufficient to check, whenever a hash table grows due to Line~\ref{alg:partial_cover:v_votes}, whether it now has size $c \ln n$, and set $S^*$ to that set if so. Therefore, we can satisfy Lines~\ref{alg:partial_cover:find_largest_set} and \ref{alg:partial_cover:vote_threshold} in time proportional to the number of hash table insertions, which is no more than $\sum_{v \in K} \FreqOf(v)$.\\
    \\
    It follows that \PartialCover runs in $O(\sum_{v \in K} \FreqOf(v))$ time.
\end{proof}

\subsection{$\alpha$-Navigable Graphs in Batches}\label{sec:batch_cover}

Having established that \PartialCover covers at least half of the elements in $U$ when given a valid parameter $\opt$, we now describe how we use it to build a full $\alpha$-navigable graph. It is easy to see that if we can cover a constant fraction of $U$, then we can do so again to cover another constant fraction of the remainder, and that we will eventually cover all of $U$ after at most $\log_2(|U|)$ calls to \PartialCover. Our full algorithm follows this general strategy.

\begin{algorithm}[htb]
    \caption{\BatchCover($P, \alpha, \varepsilon, C, c, \gamma$)}
    \label{alg:batch_cover}
    \begin{algorithmic}[1]
        \Require A set $P$ of $n$ points and constants $\alpha \geq 1$, $\varepsilon, C, c > 0$, and integer $\gamma \geq 0$.
        \Ensure A graph $G$.
        \algrule
        \State $G, U \gets$ an array of $n$ sets
        \State $Q \gets P$
        \For{$\opt \in \{2^1, 2^2, 2^3, \dots\}$}
            \For{$p_i \in Q$}
                \State $G[i] \gets$ a set of $\opt \ln n$ elements sampled from $P$\label{alg:batch_cover:random_sets}
                \State $U[i] \gets P \setminus \{p_i\} \setminus \bigcup_{p_j \in G[i]} S_j$
            \EndFor
            \For{$t \in [\log_2 n + \gamma]$}
                \For{$p_i \in Q$}
                    \State $G[i] \gets G[i] \cup \PartialCover(\{S_{i \to j}^{\alpha(1 + \varepsilon)} : j \in [n]\}, U[i], \opt, C, c)$
                \EndFor
                \State $\mathcal{N} \gets \VerifyNav(G, \alpha, \varepsilon)$
                \For{$p_i \in Q$}
                    \State $U[i] \gets \{t : (i, t) \in \mathcal{N}\}$
                    \If{$U[i] = \varnothing$}
                        $Q \gets Q \setminus \{p_i\}$
                    \EndIf
                \EndFor
            \EndFor
            \If{$Q = \varnothing$}
                \State \Return $G$
            \EndIf
        \EndFor
    \end{algorithmic}
\end{algorithm}

However, it is not so simple. If we wish to run \PartialCover on the remaining uncovered elements, we need to identify which elements of $U$ are still uncovered. This is difficult because, whereas we have efficient access to the sets containing a given element, we do not have efficient access to the elements contained in a given set. The naive approach of iterating over $U$ and removing any elements covered by some set in $\mathcal{C}$ takes too long ($\Omega(n \cdot |\mathcal{C}|)$ time).

This is where we employ the \VerifyNav procedure described in \Cref{subsec:verify_nav}. Recall that given a graph, \VerifyNav takes $O(n^\omega \cdot \log(\Delta) / \varepsilon)$ time to return a list $\mathcal{N}$ of pairs $(s, t)$ that includes any pair whose $\alpha$ constraint is not satisfied and excludes any pair whose $\alpha(1 + \varepsilon)$ constraint is satisfied.

Thus, we can use \VerifyNav to identify all $U$ that remain uncovered for a given set cover instance. Upon obtaining a partial set cover $\mathcal{C}$ for the set cover instance corresponding to point $p_i$, we add to $G$ the edges $(p_i, p_j)$ for all $S_{i \to j} \in \mathcal{C}$. Then, we run \VerifyNav on $G$ to obtain $\mathcal{N}$, and the remaining uncovered points are all those $t$ where $(p_i, t) \in \mathcal{N}$. But note that using \VerifyNav in this way to materialize the remainder of $U$ for only a single set cover instance is wasteful. Instead, we run a single \VerifyNav to materialize the uncovered points for all set cover instances simultaneously.

That is, we execute \PartialCover on each set cover instance, obtaining a partial cover for each. We add the edges corresponding to those covers to an accumulated graph. Then, we run \VerifyNav on the graph to list all remaining unsatisfied constraints. We run \PartialCover again on these remaining constraints. This process repeats roughly $O(\log n)$ times, which we know to be sufficient because each round covers a constant fraction of constraints. An important caveat with this approach is that \PartialCover only succeeds in covering enough constraints when given a parameter $\opt$ close enough to its true optimum cover size $\OPT_{\alpha(1 + \varepsilon)}$. To accomplish this, we run the entire procedure over doubling guesses of $\OPT_{\alpha(1 + \varepsilon)}$, ensuring that we find a serviceable value $\opt$ within $\log(\OPT_{\alpha(1 + \varepsilon)})$ attempts.

Only one detail remains. Recall from \Cref{lem:partial_cover_running_time} that the running time of \PartialCover is proportional to the sum of all frequencies of the elements it samples. It is not hard to see that, across all set cover instances, the frequency of an element is large (roughly $n/2$) on average. Thus, we attach an additional preprocessing step alongside each guess $\opt$ which serves to bound the frequency of each element. Namely, we add $\opt \ln n$ random edges to each point. This covers a large number of constraints -- in particular, it covers all constraints with high frequency. After having added these random edge, we use \VerifyNav to identify the targets that remain uncovered, all of which have low frequencies.

This full algorithm, \BatchCover, is laid out in \Cref{alg:batch_cover}. It remains to be shown that \BatchCover returns a near-optimally-sparse $\alpha$-navigable graph, and that it does so in $\widetilde O(n^\omega \log(\Delta)/\varepsilon)$ time.

\subsubsection{Approximation factor}

Given that \PartialCover covers a constant fraction of the unsatisfied constraints for a given set cover instance, it is easy to see that \BatchCover must successfully solve that set cover instance on or before the first iteration where $\opt$ exceeds the true optimal cover size for that instance.

\begin{observation}\label[observation]{obs:batch_cover_succeeds}
    Given sufficiently-large constants $C$ and $c$, as well as constant $\opt$ satisfying $\OPT_{\alpha(1 + \varepsilon)} < \opt \leq 2\OPT_{\alpha(1 + \varepsilon)}$, \BatchCover (\Cref{alg:batch_cover}) covers $U$ in no more than $\log n + \gamma$ calls to \PartialCover with probability $\geq 1 - O(1/n^{\gamma - 1})$, for any constant $\gamma$.
\end{observation}

\begin{proof}
    If a call to \PartialCover covers at least half of $U$, we'll say this call was ``successful''. By \Cref{lem:batch_cover_fraction}, each call to \PartialCover is successful with probability at least $1 - 1/n$ given adequate constants $C$ and $c$. All points in $U$ are covered after no more than $\lceil \log_2 n\rceil + 1$ successful calls. Thus, if we call \PartialCover $\lceil \log_2 n\rceil + \gamma$ times, we can only fail to cover $U$ if the unsuccessful calls number at least $\gamma$.

    By a union bound over the possible sets of $\gamma$ unsuccessful calls, the failure probability of \BatchCover is at most
    $$
        \binom{\lceil \log_2 n\rceil + \gamma}{\gamma}
        \left(\frac{1}{n}\right)^\gamma
        \le
        \left(\frac{e(\lceil \log_2 n\rceil + \gamma)}{\gamma n}\right)^\gamma.
    $$
    In particular, for any constant $\gamma$, this failure probability is
    $$
        O\left(\frac{1}{n^{\gamma - 1}}\right).
    $$
\end{proof}

A bound on the size of the adjacency list constructed for a given point follows directly from \Cref{obs:partial_cover_output_size}.

\begin{lemma}\label{lem:batch_cover_approximation_factor}
    For a given point whose set cover instance has optimal cover size $\OPT_{\alpha(1 + \varepsilon)}$, \BatchCover constructs an adjacency list of size $O(\OPT_{\alpha(1 + \varepsilon)} \log n)$ with probability $\geq 1 - O(1/n^{\gamma - 1})$.
\end{lemma}

\begin{proof}
    Because \Cref{alg:batch_cover} processes $\opt$ in increasing guesses of doubling magnitude, there must be some iteration where $\OPT_{\alpha(1 + \varepsilon)} < \opt \leq 2 \OPT_{\alpha(1 + \varepsilon)}$. If $p_i$ is still unsolved (i.e. contained in $Q$) by this point, then by \Cref{obs:batch_cover_succeeds}, it will be solved on this iteration with probability at least $1 - 1/n^{\gamma - 1}$. Thus, by \Cref{obs:partial_cover_output_size}, with probability $\geq 1 - 1/n^{\gamma - 1}$, $p_i$ is removed from $Q$ having received at most $\frac{C}{c} \cdot \opt \cdot (\log_2 n + \gamma) \leq \frac{2C}{c} \cdot \OPT_{\alpha(1 + \varepsilon)} \cdot (\log_2 n + \gamma)$ edges from \PartialCover.

    Thus, the final adjacency list of $p_i$ is the union of $\opt \ln n \leq 2\OPT_{\alpha(1 + \varepsilon)} \ln n$ random sets assigned by Line~\ref{alg:batch_cover:random_sets} and $\leq \frac{2C}{c} \cdot \OPT_{\alpha(1 + \varepsilon)} \cdot (\log_2 n + \gamma)$ edges returned by \PartialCover. Combining these yields a total of $O(\OPT_{\alpha(1 + \varepsilon)} \log n)$ edges.
\end{proof}

Applying a union bound over all $n$ set cover instances, all adjacency lists are within an $O(\log n)$ factor of optimal with probability $\geq 1 - 1/n^{\gamma - 2}$.

\subsubsection{Running Time}

We now consider the running time of \Cref{alg:batch_cover}.

While \Cref{lem:batch_cover_fraction} establishes that \BatchCover successfully solves each set cover instance within $\log(\OPT)$ attempts, efficiently executing this procedure requires carefully bounding the time spent processing votes. Recall from \Cref{lem:partial_cover_running_time} that \PartialCover spends time proportional to the sum frequency of all sampled elements. To achieve our desired running time, introduce a randomized preprocessing step to each set cover instance to limit the maximum number of sets containing any uncovered element.

\begin{lemma}[Frequency Bounding via Random Sets]\label{lem:freq_bounding}
    Suppose we add a collection $\mathcal{R}$ of $\opt \ln n$ sets chosen uniformly at random to the cover for $p_i$. Let $U$ be the set of elements remaining uncovered. With high probability, no element in $U$ is contained in more than $O(n / \opt)$ sets.
\end{lemma}

\begin{proof}
    Consider an element $p$ that is contained in at least $\lambda \cdot n / \opt$ sets for some constant $\lambda$. The probability that a single randomly chosen set covers $p$ is at least $\frac{\lambda \cdot n / \opt}{n} = \frac{\lambda}{\opt}$. If we sample $\opt \ln n)$ random sets, the probability that $p$ is not covered by any of them is at most
    $$\left(1 - \frac{\lambda}{\opt}\right)^{\opt \ln n} \leq \exp(-\lambda \ln n) \leq \frac{1}{n^\lambda}$$
    for an arbitrarily large constant $\lambda$. Taking a union bound over all elements ensures that all high-frequency elements are covered with probability $\geq 1 - 1/n^{\lambda - 1}$, leaving only elements contained in $O(n / \opt)$ sets in $U$.
\end{proof}

With the frequency of the uncovered elements bounded, we can efficiently execute the voting process inside \PartialCover. Since \BatchCover adds $\opt \ln n$ outneighbors to each point $p_i$, it follows that with probability $\geq 1 - 1/n^{\lambda - 2}$ (union bounding across all set cover instances), any uncovered target in any instance is contained in fewer than $\lambda \cdot n / \opt$ sets. We now apply this directly to \Cref{lem:partial_cover_running_time} to achieve the following result:

\begin{corollary}\label[corollary]{cor:partial_cover_random_edges}
    The execution of \PartialCover once for each active set cover instance in \BatchCover takes $\widetilde O(n^2)$ total time.
\end{corollary}

Finally, we can bound the total running time of \BatchCover.

\begin{lemma}\label{lem:batch_cover_running_time}
    The execution of \BatchCover takes $\widetilde O(n^\omega \log(\Delta) / \varepsilon)$ time.
\end{lemma}

\begin{proof}
    From \Cref{cor:partial_cover_random_edges}, we know that one execution of \PartialCover on each set cover instance, after having bounded the frequency each target by means of adding random edges, takes $O(n^2)$ time. We follow this with a call to \VerifyNav, which takes $O(n^\omega \log(\Delta) / \varepsilon)$. This dominates the $O(n^2)$ time for the calls to \PartialCover, as $\omega$ is necessarily at least 2. We perform this combination of these \PartialCover and \VerifyNav up to $\log n + \gamma$ times for a given choice of $\opt$, ensuring that any active set cover instance $s$ where $\opt$ is sufficiently close to $\OPT_{\alpha(1 + \varepsilon)}(s)$ completes. Finally, we run this on up to $\log n$ values of $\opt$ before guaranteeing that $\opt$ is sufficiently close to $\OPT_{\alpha(1 + \varepsilon)}(s)$ for every $s$ at least once. Therefore, the total running time of \BatchCover is $O(n^\omega (\log(\Delta) / \varepsilon)\log^2 n)$.
\end{proof}

The combination of \Cref{lem:batch_cover_running_time} and \Cref{lem:batch_cover_approximation_factor} yields \Cref{thm:upper_bound}.

\subsection{Extension to \(\tau\)-Monotonic Graphs}
\label{subsec:tau_monotonic}

Although we have presented our algorithm in terms of $\alpha$-navigability, the underlying set cover and batched verification framework is not specific to multiplicative progress. We briefly observe that the preceding framework also applies to the construction of sparse $\tau$-monotonic graphs \cite{PengChoiChan:2023,ConwayEtAl2025}. For \(\tau > 0\), a directed graph \(G=(P,E)\) is \emph{\(\tau\)-monotonic} if, for every pair of distinct points \(s,t\in P\), either \((s,t)\in E\), or there exists an edge \((s,v)\in E\) such that
\[
    d(v,t)<d(s,t)-\tau.
\]
Thus, each step makes additive, rather than multiplicative, progress toward the target. The direct-edge alternative is necessary when \(d(s,t)\leq\tau\), since in this case no point can satisfy the displayed inequality.

As with \(\alpha\)-navigability, selecting the out-neighbors of each source \(s\) is a set cover problem. For each potential edge \((s,v)\), define
\[
    T_{s\to v}
    =
    \{t\in P\setminus\{s\}:
      v=t \text{ or } d(v,t)<d(s,t)-\tau\}.
\]
The \MemberOf, \FreqOf, and \SetOf operations from \Cref{subsec:set_access_model} continue to be supported in constant time after \(O(n^2\log n)\) preprocessing. In particular, for a fixed source \(s\) and target \(t\), the points \(v\) satisfying $d(v,t)<d(s,t)-\tau$ form a prefix of the points sorted by distance from \(t\). The length of this prefix can be found by binary search and stored for every pair \((s,t)\). The possible additional set \(T_{s\to t}\), corresponding to the direct edge \((s,t)\), is handled separately.

The other change is that distance discretization is additive. By rescaling the distance function and \(\tau\), assume without loss of generality that the minimum nonzero pairwise distance is \(1\), so the maximum pairwise distance is \(\Delta\). We use the additive scales $r_\ell=\ell\epsilon$ for $\ell=0,\ldots,\left\lceil\frac{\Delta}{\epsilon}\right\rceil$
and round each distance \(x\) down to $\rho(x)=\epsilon\left\lfloor\frac{x}{\epsilon}\right\rfloor$.
Consequently,
\[
    x-\epsilon<\rho(x)\leq x,
\]
and there are \(O(\Delta/\epsilon)\) distance scales.

The corresponding \VerifyNav procedure performs one matrix multiplication for each scale. At scale \(r_\ell\), in the notation of \Cref{lem:verify_nav:single_scale}, we define
\[
    B^{(\ell)}_{j,k}
    =
    \begin{cases}
        1, & j=k \text{ or } d(p_j,p_k)<r_\ell-\tau,\\
        0, & \text{otherwise}.
    \end{cases}
\]
The diagonal entries account for direct edges. For a pair \((s,t)\) whose distance rounds to \(r_\ell\), we report the constraint as unsatisfied if the corresponding entry of \(AB^{(\ell)}\) is zero. If \(\mathcal{U}_\eta(G)\) denotes the set of unsatisfied \(\eta\)-monotonicity constraints, the resulting set \(\mathcal{N}_\tau(G)\) satisfies
\[
    \mathcal{U}_\tau(G)
    \subseteq
    \mathcal{N}_\tau(G)
    \subseteq
    \mathcal{U}_{\tau+\epsilon}(G).
\]
Indeed, rounding down cannot cause an unsatisfied \(\tau\)-constraint to be removed, while any constraint satisfied with additive progress \(\tau+\epsilon\) remains satisfied after its source-target distance is rounded down by less than \(\epsilon\). The procedure uses \(O(\Delta/\epsilon)\) matrix multiplications and therefore runs in \(\widetilde O(n^\omega\Delta/\epsilon)\) time.

We may now apply \BatchCover without further modification.

\begin{corollary}
\label{cor:tau_monotonic}
Let \(P\) be a set of \(n\) points with aspect ratio \(\Delta_P\), after normalizing its minimum nonzero pairwise distance to \(1\). For every \(\tau\geq 0\) and \(0<\epsilon\leq 1\), there is a randomized algorithm running in
\[
    \widetilde O\left(n^\omega\frac{\Delta_P}{\epsilon}\right)
\]
time that constructs a \(\tau\)-monotonic graph \(G\) such that, with high probability, every \(s\in P\) has out-degree at most $O(\OPT_{\tau+\epsilon}(s) \log n)$, where \(\OPT_{\tau+\epsilon}(s)\) is the minimum out-degree required at \(s\) in any \((\tau+\epsilon)\)-monotonic graph. Consequently,
\[
    |E|
    \leq
    O(\log n)\cdot\OPT_{\tau+\epsilon}.
\]
\end{corollary}

\begin{proof}
The analysis of \BatchCover uses only constant-time access to \MemberOf, \FreqOf, and \SetOf, together with a verification procedure whose output lies between the constraints unsatisfied at the desired parameter and those unsatisfied at the resource-augmented parameter. The set-access construction and the additive version of \VerifyNav above provide these properties for parameters \(\tau\) and \(\tau+\epsilon\), respectively. The remainder of the proof of \Cref{thm:upper_bound} applies unchanged.
\end{proof}
The change from a logarithmic to linear dependency on $\Delta_P$ in the running time when transitioning from $\alpha$-navigability to $\tau$-monotonicity is not surprising, as this reflects the increase in number of relevant distance scales from $O(\log_\alpha(\Delta_P))$ to $O(\Delta_P/\tau)$. In contrast to $\alpha$-navigable graphs, which guarantee a logarithmic search convergence of $O(\log_\alpha(\Delta_P))$ steps, $\tau$-monotonic graphs likewise inherit a linear bound of $O(\Delta_P/\tau)$ steps to converge.

\section{Hardness}

This section addresses two natural questions about the algorithms from the previous sections.

\begin{enumerate}
    \item Can the running time of the resource augmented algorithm be improved?
    The main bottleneck is the subroutine that, at a fixed distance scale \(d_0\), lists all source-target pairs \((s,t)\) whose \(\alpha\)-navigability constraints become satisfied after the newly added edges are inserted. In our current implementation, this subroutine is based on fast matrix multiplication and runs in \(O(n^\omega)\) time. Is there a substantially faster implementation?

    \item Can one eliminate the discretization step and obtain a more efficient exact single-criterion algorithm?
    More precisely, can one efficiently test, in the original metric, which \(\alpha\)-navigability constraints are satisfied by the current edge set? Any \(O(n^{2.5-\varepsilon})\)-time algorithm for this task, for a fixed constant \(\varepsilon > 0\), would improve upon the current worst-case running time of the best known single-criterion algorithm.
\end{enumerate}

We formalize the corresponding testing problems and prove hardness results that answer both questions negatively.

\begin{definition}
Let \(G = (P,E)\) be a directed graph, let \(d\) be a metric on \(P\), let \(\alpha > 1\), and let \(d_0 > 0\).
The \emph{single-scale \(\alpha\)-navigability testing problem} asks to report all pairs \((s,t) \in P \times P\) such that \(d(s,t) = d_0\) and there exists \(u \in P\) with
\[
(s,u) \in E
\qquad\text{and}\qquad
d(u,t) \le \frac{1}{\alpha} d(s,t).
\]
\end{definition}
Question (1) asks whether one can avoid the bottleneck of matrix multiplication in solving this problem.
\Cref{lem:BMMReduction} answers this question in the negative, showing that a reduction exists in the opposite direction, and that this problem is essentially equivalent to Boolean matrix multiplication.
\begin{definition}
Let \(G = (P,E)\) be a directed graph, let \(d\) be a metric on \(P\), and let \(\alpha > 1\).
The \emph{\(\alpha\)-navigability testing problem} asks to report all pairs \((s,t) \in P \times P\) such that there exists \(u \in P\) with
\[
(s,u) \in E
\qquad\text{and}\qquad
d(u,t) \le \frac{1}{\alpha} d(s,t).
\]
\end{definition}
Question (2) asks whether any $O(n^{2.5-\varepsilon})$-time implementation
exists, as such an implementation would improve upon the state of
the art single-criterion approximation.
\Cref{cor:APSP-hardness} answers this question
in the negative by demonstrating APSP-hardness.

\begin{lemma}\label{lem:BMMReduction}
There is a reduction from Boolean matrix multiplication on two \(n \times n\) matrices to single-scale \(\alpha\)-navigability testing that uses one oracle call and $O(n^2)$ additional time.
\end{lemma}

\begin{proof}
Let \(A,B \in \{0,1\}^{n \times n}\). We construct an instance of single-scale \(\alpha\)-navigability testing as follows.

Let \(P = S \uplus U \uplus T\), where \(S = \{s_1,\dots,s_n\}\), \(U = \{u_1,\dots,u_n\}\), and \(T = \{t_1,\dots,t_n\}\).

We first define a metric \(d\) on \(P\). Let \(H\) be the undirected weighted graph on vertex set \(P\) containing the following edges:
\begin{itemize}
    \item for every \(i,j \in [n]\), an edge \(\{s_i,t_j\}\) of weight \(1\);
    \item for every \(k,j \in [n]\) with \(B_{k,j} = 1\), an edge \(\{u_k,t_j\}\) of weight \(1/\alpha\).
\end{itemize}
Let \(d\) be the shortest-path metric of \(H\).

Next, define a directed graph \(G = (P,E)\) by adding the edge \((s_i,u_k)\) to \(E\) if and only if \(A_{i,k} = 1\).

We claim that one call to the single-scale \(\alpha\)-navigability testing oracle with parameter \(d_0 = 1\) suffices to compute the Boolean product \(AB\).

First observe that in $H$, for every \(i,j \in [n]\), the direct edge \(\{s_i,t_j\}\) has weight \(1\). This edge realizes the shortest-path distance, since alternative path from \(s_i\) to \(t_j\) has length greater than 1, as there are no direct edges between $S$ and $U$. Hence, \(d(s_i,t_j) = 1\).

Therefore, for every \(i,j \in [n]\), the oracle reports \((s_i,t_j)\) if and only if there exists \(k \in [n]\) such that
\[
(s_i,u_k) \in E \iff A_{i,k}=1
\qquad\text{and}\qquad
d(u_k,t_j) \le \frac{1}{\alpha}\cdot d(s_i,t_j) = \frac{1}{\alpha}.
\]
Further note, for every \(k,j \in [n]\),
\[
d(u_k,t_j) \le \frac{1}{\alpha}
\qquad\Longleftrightarrow\qquad
B_{k,j} = 1,
\]
because the only possible path in $H$ of weight $\leq\frac{1}{\alpha}$ between $u_k$ and $t_j$ in $H$, is the direct edge $(u_k,t_j)$ which is present if and only if $B_{k,j}=1$.
Overall, a pair \((s_i,t_j)\) is reported if and only if there exists \(k \in [n]\) such that \(A_{i,k}=B_{k,j}=1\).
\end{proof}

As a first step towards the proof of \cref{cor:APSP-hardness} we define the problem of Boolean min-plus product and prove that it has equivalent running time to $\alpha$-navigability testing up to polylogarithmic factors.
\begin{definition}
Let \(M\) be an \(n \times n\) matrix over \(\mathbb{R}_{\ge 0}\), and let \(N\) be an \(n \times n\) Boolean matrix.
Their \emph{Boolean min-plus product} is the matrix \(M \star N\) defined by
\[
[M \star N]_{i,j}
    = \min \{ M_{i,k} : N_{k,j} = 1 \},
\]
with the convention that the minimum over the empty set is \(+\infty\).
\end{definition}
Equivalently, one may regard \(N\) as a matrix over \(\{0,\infty\}\), replacing each \(1\) by \(0\) and each \(0\) by \(\infty\), and then take the usual min-plus product.

This is a problem that has been used as a subroutine to compute node-weighted APSP in previous works by Chan, Yuster, and Abboud et al most recently (see "All-Pairs Shortest Paths with Few Weights per Node
" for more details). It's hardness was studied in the recent "Universe Reduction for APSP: Equivalence of Three Fine-Grained Hypotheses", which demonstrated APSP-hardness.

\begin{lemma}\label{lem:boolean-min-plus-hard}[Lemma 6.3 from Universe Reduction for APSP: Equivalence of Three Fine-Grained Hypotheses]
For every \(\varepsilon>0\), there does not exist an \(O(n^{2.5-\varepsilon})\) algorithm for Boolean min-plus product on \(n\times n\) matrices, unless the APSP conjecture is falsified.
\end{lemma}

For the purpose of proving $\alpha$-navigability testing is equivalent to Boolean min-plus product computation, we define an intermediate problem.
\begin{definition}
Let \(A \in \{0,1\}^{n \times n}\) and let \(B,C \in \mathbb{R}_{\ge 0}^{n \times n}\).
The \emph{Boolean min-plus threshold testing problem} asks to report all pairs \((i,j) \in [n] \times [n]\) such that
\[
\min \{ B_{k,j} : A_{i,k} = 1 \} \le C_{i,j},
\]
again with the convention that the minimum over the empty set is \(+\infty\).
\end{definition}

\begin{lemma}\label{lem:bmptt-equivalence}
Boolean min-plus product and Boolean min-plus threshold testing are equivalent up to polylogarithmic overhead. More precisely, if one of the two problems on \(n \times n\) matrices can be solved in time \(T(n)\), then the other can be solved in time \(\widetilde O(T(n) + n^2)\).
\end{lemma}

\begin{proof}
We first reduce Boolean min-plus product to threshold testing.

Let \(M \in \mathbb{R}_{\ge 0}^{n \times n}\) and \(N \in \{0,1\}^{n \times n}\), our goal is to compute \(M \star N\).
For each row \(i\), replace the entries of \(M\) by their ranks within that row. Formally, define \(R \in [n]^{n \times n}\) by
\[
R_{i,k} = \bigl| \{k' \in [n] : M_{i,k'} \le M_{i,k} \} \bigr|.
\]
Then, for every \(i,j\), the value \([R \star N]_{i,j}\) is exactly the rank, within row \(i\), of the minimum entry among the columns \(k\) with \(N_{k,j}=1\), which makes recovering \(M\star N\) from \(R\star N\) straight forward in \(\widetilde{O}(n^2)\) time.

It therefore suffices to compute \(R \star N\), and this can be done using \(O(\log n)\) calls to the threshold-testing oracle by parallel binary search on all entries simultaneously. More precisely, a threshold matrix $C$ is formed, and a search space for each coordinate is maintained (initialized to \([n]\)). At each step, we query threshold testing on \((N^T,R^T,C^T)\). 
Following this call, for each pair \(i,j\) we learn whether:
\[
\min \{ [R^T]_{k,i}=R_{i,k} : [N^T]_{j,k}=N_{k,j} = 1 \} \le [C^T]_{j,i} \iff [R\star N]_{i,j}\leq C_{i,j}
\]
Based on the answer, for each \(i,j\) we update the search space and \(C_{i,j}\) accordingly before the next query. Observe that the entries of \([R\star N]\) are in \([n]\cup \{\infty\}\) and thus the searches converge in \(log(n)\) steps.

Conversely, suppose we are given an oracle for Boolean min-plus product, and let \(A \in \{0,1\}^{n \times n}\) and \(B,C \in \mathbb{R}_{\ge 0}^{n \times n}\) be an instance of threshold testing. Compute \(D = (B^T \star A^T)^T\).
Then, for every \(i,j\),
\[
D_{i,j}
    = \min \{ B_{k,j} : A_{i,k} = 1 \}.
\]
Hence we can obtain the answer by comparing \(D_{i,j}\) with \(C_{i,j}\) entrywise.
\end{proof}

We next relate threshold testing to \(\alpha\)-navigability testing.

\begin{lemma}\label{lem:threshold-to-nav}
Boolean min-plus threshold testing and \(\alpha\)-navigability testing have equivalent running times - if one of the two can be solved in \(T(n)\) time, the other can be solved in \(T(n)+O(n^2)\) time.
\end{lemma}

\begin{proof}
We begin with the reduction from \(\alpha\)-navigability testing to Boolean min-plus threshold testing.

Let \(G = (P,E)\) be a directed graph on \(P = \{v_1,\dots,v_n\}\), and let \(d\) be a metric on \(P\).
Define matrices \(A \in \{0,1\}^{n \times n}\) and \(B,C \in \mathbb{R}_{\ge 0}^{n \times n}\) by
\[
A_{i,k} = 1 \iff (v_i,v_k) \in E,
\qquad
B_{k,j} = d(v_k,v_j),
\qquad
C_{i,j} = \frac{1}{\alpha} d(v_i,v_j).
\]
Then, for every \(i,j\), we have \(\min \{ B_{k,j} : A_{i,k} = 1 \} \le C_{i,j}\) if and only if there exists \(k\) such that
\[
(v_i,v_k) \in E
\qquad\text{and}\qquad
d(v_k,v_j) \le \frac{1}{\alpha} d(v_i,v_j),
\]
which is exactly \(\alpha\)-navigability constraint for $(v_i,v_j)$.

For the converse direction, let \(A \in \{0,1\}^{n \times n}\) and \(B,C \in \mathbb{R}_{\ge 0}^{n \times n}\) be an instance of Boolean min-plus threshold testing. Let \(S = \{s_1,\dots,s_n\}\), \(U = \{u_1,\dots,u_n\}\), and \(T = \{t_1,\dots,t_n\}\), and set \(P = S \uplus U \uplus T\).

Define a directed graph \(G = (P,E)\) by: \((s_i,u_k) \in E\) if and only if \(A_{i,k} = 1\).

We now define a metric on \(P\). Let \(W > \alpha \max_{i,j} C_{i,j} + \max_{k,j} B_{k,j}\).
Let \(H\) be the undirected weighted graph on \(P\) containing the following edges:
\begin{itemize}
    \item for every \(k,j \in [n]\), an edge \(\{u_k,t_j\}\) of weight \(W + B_{k,j}\);
    \item for every \(i,j \in [n]\), an edge \(\{s_i,t_j\}\) of weight \(\alpha \cdot (W +  C_{i,j})\).
\end{itemize}
Let \(d\) be the shortest-path metric of \(H\).

We claim that, for all \(k,j\), \(d(u_k,t_j) = W + B_{k,j}\), and, for all \(i,j\), \(d(s_i,t_j) = \alpha W + \alpha C_{i,j}\).
Indeed, each of these values is realized by a direct edge. Any alternative path from \(u_k\) to \(t_j\) uses at least two edges between \(U\) and \(T\) or two edges between \(S\) and \(T\) and therefore has weight at least \(2W>W + B_{k,j}\). Likewise, any alternative path from \(s_i\) to \(t_j\) has at least one edge connecting \(S\) and \(T\) and at least one edge connecting \(U\) and \(T\) and therefore has weight at least \(\alpha\cdot W + W > \alpha\cdot (W + C_{i,j})\).

Therefore, for every \(i,k,j\),
\[
d(u_k,t_j) \le \frac{1}{\alpha} d(s_i,t_j)
\iff
W + B_{k,j} \le W + C_{i,j}
\iff
B_{k,j} \le C_{i,j}.
\]
It follows that the \(\alpha\)-navigability oracle reports \((s_i,t_j)\) if and only if there exists \(k\) such that
\[
A_{i,k} = 1
\qquad\text{and}\qquad
B_{k,j} \le C_{i,j},
\]
which is exactly the threshold-testing condition.
\end{proof}

Combining \cref{lem:bmptt-equivalence,lem:threshold-to-nav,lem:boolean-min-plus-hard} yields the following corollary.

\begin{corollary}\label{cor:APSP-hardness}
For every constant \(\varepsilon>0\), there does not exist an \(O(n^{2.5-\varepsilon})\) algorithm for \(\alpha\)-navigability testing unless the APSP conjecture is falsified.
\end{corollary}

\section{Acknowledgments}
This material is based upon work supported by the National Science Foundation under Award Nos. 2403235, 2504470, 2551614, 2551615, and 2551617. Any opinions, findings, and conclusions or recommendations expressed in this material are those of the authors and do not necessarily reflect the views of the National Science Foundation.

\bibliographystyle{plain}
\bibliography{references}
\end{document}